\documentclass[conference,twocolumn]{IEEEtran}

\usepackage{epsfig}
\usepackage{graphicx}
\usepackage{color}

\usepackage{amsfonts}
\usepackage{latexsym}
\usepackage{amsmath}
\usepackage{amssymb}

\newcommand{\mikel}[1]{\marginpar{*}{\bf Mikel's remark}: {\em #1}}
\newcommand{\sidj}[1]{\marginpar{+}{\bf Sid's remark}: {\em #1}}
\newcommand{\bikd}[1]{\marginpar{=}{\bf Bikash's remark}: {\em #1}}
\newcommand{\suppress}[1]{}

\newtheorem{theorem}{Theorem}
\newtheorem{lemma}{Lemma}[section]
\newtheorem{claim}{Claim}[section]

\newcommand{\bydef}{\stackrel{\triangle}{=}}

\newcommand{\poly}{\mbox{poly}}

\newcommand{\charac}{{p}}

\newcommand{\mess}{u}
\newcommand{\ind}{r}

\newcommand{\F}{\mathbb{F}}

\newcommand{\Ca}{{C}}

\newcommand{\ad}{{\tt{add}}}
\newcommand{\om}{{\tt{omni}}}

\def\conf{\mbox{$\cal{S}$}}

\newcommand{\defn}{\stackrel{\triangle}{=}}

\def\cA{\mbox{$\cal{A}$}}

\def\cX{\mbox{$\cal{X}$}}

\def\cC{\mbox{$\cal{C}$}}

\def\cI{{\mbox{${\cal{I}}$}}}

\def\cV{\mbox{$\cal{V}$}}
\def\cE{\mbox{$\cal{E}$}}
\def\cU{\mbox{$\cal{U}$}}

\def\e{\varepsilon}

\def\bx{{\bf x}}
\def\by{{\bf y}}

\def\bU{{\bf U}}
\def\bX{{\bf X}}
\def\bJ{{\bf J}}

\def\bl{{n}}
\def\rate{{R}}
\newcommand{\Graph}{{{\cal G}}}

\def\01{\{0,1\}}

\newcommand{\remove}[1]{}

\begin{document}

\title{Binary Causal-Adversary Channels}
\author{
\authorblockN{M. Langberg}
\authorblockA{Computer Science Division  \\
Open University of Israel \\
Raanana 43107, Israel \\
{\tt mikel@openu.ac.il}\vspace*{-4.0ex}} \and
\authorblockN{S. Jaggi}
\authorblockA{Department of Information Engineering\\
Chinese University of Hong Kong \\
Shatin, N.T., Hong Kong \\
{\tt jaggi@ie.cuhk.edu.hk}\vspace*{-4.0ex}} \and
\authorblockN{B. K. Dey}
\authorblockA{Department of Electrical Engineering  \\
Indian Institute of Technology Bombay \\
Mumbai, India, 400 076 \\
{\tt bikash@ee.iitb.ac.in}\vspace*{-4.0ex}}
}

\date{}
\maketitle
%
\footnotetext{The work of B. K. Dey was supported by Bharti Centre for Communication in IIT Bombay, that of M. Langberg was supported in part by ISF grant 480/08, and that of S. Jaggi was partially supported by MS-CU-JL grants.}

\begin{abstract}
In this work we consider the communication of information in the presence of a
{\em causal} adversarial jammer.
In the setting under study, a sender wishes to communicate a message to a receiver by transmitting a codeword $\bx=(x_1,\dots,x_\bl)$ bit-by-bit over a communication channel. The adversarial jammer can view the transmitted bits $x_i$ one at a time, and can change up to a $p$-fraction of them.
However, the decisions of the jammer must be made in an {\em online} or {\em causal} manner.
Namely, for each bit $x_i$ the jammer's decision on whether to corrupt it or not (and on how to change it) must depend only on $x_j$ for $j \leq  i$. This is in contrast to the ``classical'' adversarial jammer which may base its decisions on its complete knowledge of $\bx$.
We present a non-trivial upper bound on the amount of information that can be communicated.
We show that the achievable rate can be asymptotically
no greater than $\min\{1-H(p),(1-4p)^+\}$. Here $H(.)$ is the binary entropy function, and $(1-4p)^+$ equals $1-4p$ for $p\leq 0.25$, and $0$ otherwise.




\end{abstract}
%


\section{Introduction}

Consider the following adversarial communication scenario. A sender Alice wishes to transmit a message $\mess$ to a receiver Bob.
To do so, Alice encodes $\mess$ into a codeword $\bx$ and transmits it over a {\em binary channel}. The codeword $\bx=x_1,\dots,x_\bl$ is a binary vector of length $\bl$.
However, Calvin, a malicious adversary, can observe $\bx$ and corrupt up to a $p$-fraction of the $\bl$ transmitted bits, {\em i.e.}, $pn$ bits.

In the classical adversarial channel model, e.g., \cite{CT06}, it is usually assumed that Calvin has full knowledge of the entire codeword $\bx$, and based on this knowledge (together with the knowledge of the code shared by Alice and Bob) Calvin can maliciously plan what error to impose on $\bx$.
We refer to such an adversary as an {\em omniscient} adversary.
For binary channels, the optimal rate of communication in the presence of an omniscient adversary has been an open problem in classical coding theory for several decades. 
The best known lower bound is given by the Gilbert-Varshamov bound~\cite{gilbert,varshamov}, which implies
that Alice can transmit at rate $1-H(2p)$ to Bob.
Conversely, the tightest upper bound was given by McEliece et al.~\cite{mceliece2}, and has a positive gap from the lower bound for all $p \in (0,1/4)$ (see
Fig. \ref{fig:bounds}).

\begin{figure}[h]
\centering
\includegraphics[width=3.5in]{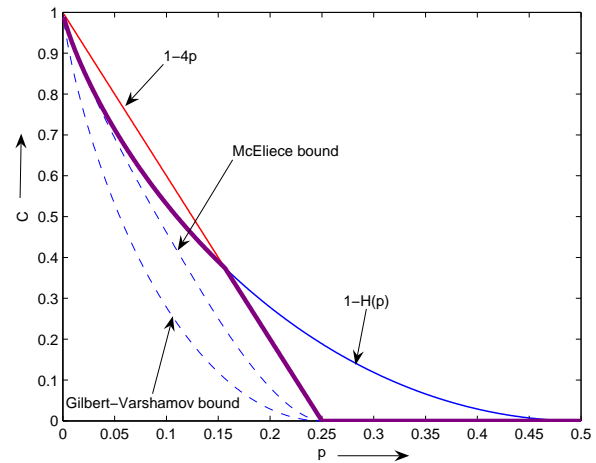}
\caption{Bounds on capacity of the adversarial channel. The bold line in purple is our upper bound of $\min\{1-H(p),(1-4p)^+\}$.}
\label{fig:bounds}
\end{figure}

In this work we initiate the analysis of coding schemes that allow communication against certain adversaries that are weaker than the omniscient adversary. We consider adversaries that behave in a {\em causal} or  {\em online} manner. Namely, for each bit $x_i$, we assume that Calvin decides whether to change it or not (and if so, how to change it) based on the bits $x_j$, for $j \leq i$ alone, {\em i.e.}, the bits that he has already observed. In this case we refer to Calvin as a {\em causal} adversary.

Causal adversaries arise naturally in practical settings, where adversaries typically have no {\em a priori} knowledge of Alice's message $\mess$. In such cases they must simultaneously learn $\mess$ based on Alice's transmissions, and jam the corresponding codeword $\bx$ accordingly. This {\em causality} assumption is reasonable for many communication channels, both wired and wireless, where Calvin is not co-located with Alice.
For example consider the scenario in which the transmission of $\bx=x_1,\dots,x_\bl$ is done during $\bl$ channel uses over time, where at time $i$ the bit $x_i$ is transmitted over the channel.
Calvin can only corrupt a bit when it is transmitted (and thus its error is based on its view so far).
To decode the transmitted message, Bob waits until all the bits have arrived.
As in the omniscient model, Calvin is restricted in the number of bits $p\bl$ he can corrupt.
This might be because of limited processing power or limited transmit energy.

Recently, the problem of codes against causal adversaries was considered and solved by the authors~\cite{djl} for {\em large-$q$ channels}, {\em i.e.}, channels where Alice's codeword $\bx=x_1,\dots,x_\bl$ is considered to be a vector
of length $\bl$ over a field of ``large'' size $q$. Each {\em symbol}
$x_i$ may represent a large packet of bits in practice. Calvin is
allowed to arbitrarily corrupt a $p$-fraction of the symbols, rather
than bits. A tight characterization of the rate-region for various
scenarios is given in~\cite{djl}, and computationally efficient
codes that achieve these rate-regions are presented. However, the
techniques used in characterizing the rate-region of causal
adversaries over large-$q$ channels do not work over binary channels.
This is because each symbol in a large-$q$ channel can contain
within it a ``small" hash that can be used to verify the symbol. This
is the crux of the technique used 
to achieve the lower bounds in~\cite{djl}. We currently do not know how to extend
this method to binary channels. Conversely, for upper bounds, the
geometry of the space of length-$\bl$ codewords over large-$q$
alphabets is significantly different than that corresponding to binary alphabets. 
For instance, for large-$q$ channels the volume of an $n$-sphere of
radius $\alpha n$ ($0\leq \alpha \leq 1$) over $F_q$ is $\sim q^{n\alpha }$,
This leads to simpler bounds for large-$q$ channels.

In this work we initiate the study of binary causal-adversary channels, and present two upper bounds on their capacity: $1-H(p)$, and $(1-4p)^+$.
The upper bound of $1-H(p)$ is very ``natural". Namely, it is not hard to verify that if Calvin attacks Alice's transmission by simulating the well-studied Binary Symmetric Channel~\cite{CT06}, he can force a communication rate of no more than $1-H(p)$.
The upper bound of $(1-4p)^+$ presented in this work is non-trivial for both
its implications and its proof techniques. The bound demonstrates
that at least for some values of $p$, the achievable rate is bounded
away from $1-H(p)$. For $p \in (p_0,0.5)$, $1-4p$ is strictly less
than $1-H(p)$ (here $p_0$ is the value of $p$ satisfying $H(p) = 4p$,
and can be computed to be approximately $0.15642\ldots$). In fact
for $p \in (0.25,0.5)$ our bound implies that no communication at
positive rate is possible, which is much stronger than the result
obtained by the upper bound of $1-H(p)$ (see Fig.~\ref{fig:bounds}).
Our proof techniques include a combination of tools from the fields of Extremal Combinatorics (e.g. Tur\'{a}n's theorem~\cite{turan}), and classical Coding Theory (e.g. the Plotkin bound~\cite{plotkin1, brouwer1}).

\section{Model}
\label{sec:defns}
For any integer $i$ let $[i]$ denote the set $\{1,\dots,i\}$.
Let $\rate \geq 0$ be Alice's {\em rate}.
An $(\bl,\rate\bl)$-{\em code} $\cC$ is defined by Alice's encoder and Bob's corresponding decoder, as below.

\noindent {\bf{Alice:}} Alice's message $\mess$ is assumed to be a random variable $\bU$ with entropy $\rate \bl$, over alphabet $\cU$. We consider two types of encoding schemes for Alice.

For {\em deterministic codes},
Alice's message $\bU$ is assumed to be uniformly distributed over $\cU = [2^{\rate \bl}]$. Her
{\em deterministic encoder} is a deterministic function $f_D(.)$ that maps every $\mess$ in
$[2^{\rate \bl}]$ to a vector $\bx(\mess) = (x_1,\dots, x_\bl)$
in $\{0,1\}^\bl$. Alice's {\em codebook} $\cX$ is the collection $\{\bx(\mess)\}$ of all possible transmitted codewords.

More generally, Alice and Bob may use {\em probabilistic codes}. For
such codes, the random variable $\bU$ 
corresponding to Alice's
message $p_U$ may have an arbitrary distribution $p_U$ (with entropy $
\rate \bl$) over an arbitrary alphabet $\cU$.  
Alice's {\em codebook}
$\cX$ is an arbitrary collection $\{\cX(\mess)\}$ of subsets of
$\{0,1\}^\bl$. For each subset $\cX(\mess) \subset \cX$, there is a
corresponding {\em codeword random variable} $\bX(\mess)$ with
{\em codeword distribution} $p_{X(\mess)}$ over $\cX(\mess)$.
For any value $\bU=\mess$ of the message, Alice's encoder choses a
codeword from $\cX(\mess )$ randomly from the distribution
$p_{X(\mess)}$. Alice's
message distribution $p_U$, codebook $\cX$, and all the codebook
distributions $p_{X(\mess)}$ are all known to both Bob and Calvin,
but the values of the random variables $\bU$ and $\bX(.)$ are
unknown to them. If $\cX(\mess) = \{\bx (\mess, r) : r \in \Lambda_\mess\}$,
then the transmitted
codeword $\bX(\bU)$ has the probability distribution given by
$\Pr [\bX(\bU) = \bx (\mess, \ind)] = p_U(u)p_{X(\mess)}(\bx(\mess,\ind))$.
Let $p$ be the overall distribution of codewords $\bx = \bx(\mess,r)$ of Alice.
It holds that $p(\bx(\mess,\ind))=p_U(u)p_{X(\mess)}(\bx)$ and $p(\bx)=\sum_{U}{p_U(u)p_{X(\mess)}(\bx)}$.

\noindent {\bf{Calvin/Channel:}}
Calvin possesses $\bl$ {\em jamming functions} $g_i(.)$ and $\bl$ arbitrary jamming random variables $\bJ_i$ that satisfy the following constraints.

\noindent {\em Causality constraint:} For each $i \in [\bl]$, the jamming function $g_i(.)$ maps $\bx^i=(x_1,\ldots,x_i)$ and $\bJ^i=(\bJ_1,\ldots,\bJ_i)$ to an element of  $\{0,1\}$.

\noindent {\em Power constraint:} The number of indices $i \in [n]$ for which the value of $g_i(.)$ equals $1$ is at most $p\bl$. That is, for all $\bx^\bl,\bJ^\bl$, $\sum_i g_i(\bx^i,\bJ^i) \leq p\bl$.

\noindent The {\em output} of the channel is the set of bits $y_i = x_i \oplus g_i(\bx^i,\bJ^i)$ for $i=1,\ldots, n$.




\noindent {\bf{Bob:}}
Bob's {\em decoder} is a (potentially) probabilistic function $h(.)$ of the received vector $\by$. It maps the vectors
$\by = (y_1,\dots y_\bl)$ in $\{0, 1\}^\bl$ to the messages in $\cU$.

\noindent {\bf{Code parameters:}}
Bob is said to make a {\em decoding error} if the message $\mess'$ he decodes differs from the message $\mess$ encoded by Alice.
The {\em probability of error} for a given message $\mess$ is defined as the probability, over Alice, Calvin and Bob's random variables, that Bob makes a decoding error.
The probability of error of the code $\cC$ is defined as the average over all
$\mess \in \cU$ of the probability of error for message $\mess$.

We define two types of rates and corresponding capacities.

The rate $\rate$ is said to be {\it weakly achievable} if for every $\e > 0$, $\delta>0$  and every sufficiently large $\bl$ there exists an $(\bl,(\rate-\delta)\bl)$-code that allows communication with probability of error at most $\e$.
The supremum over $\bl$ of the weakly achievable rates is called the {\it weak capacity} and is denoted by $\Ca^{\tt{w}}$.

The rate $\rate$ is said to be {\it strongly achievable}\footnote{This definition is motivated by the extensive literature on error exponents in information theory -- for large classes of information-theoretic problems, e.g. \cite{gallager1, csiszar1}, the probability of error of the coding scheme is required to decay exponentially in block length.} if for every $\delta>0$, $\exists \alpha >0$  so that
for sufficiently large $\bl$ there exists an $(\bl,(\rate-\delta)\bl)$-code that allows communication with probability of error at most $e^{-\alpha n}$.
The supremum over $\bl$ of the strongly achievable rates is called the {\it strong capacity} and is denoted by $\Ca^{\tt{s}}$.

\noindent {\bf Remark:} Since a rate that is strongly achievable is always weakly achievable but the converse is not true in general, $\Ca^{\tt{w}} \geq \Ca^{\tt{s}}$.

\section{Related work and our results}
To the best of our knowledge, communication in the presence of a causal adversary has not been explicitly addressed in the literature (other than our prior work for causal adversaries over large-$q$ channels). Nevertheless, we note that the model of causal channels, being a natural one, has been ``on the table'' for several decades and the analysis of the online/causal channel model appears as an open question in the book of Csisz\'{a}r and Korner \cite{csiszar1} (in the section addressing Arbitrary Varying Channels \cite{blackwell1}). Various variants of causal adversaries have been addressed in the past, for instance \cite{blackwell1,JagLHE:05,SM06,Sar08,NL08} -- however the models considered therein differ significantly from ours.

At a high level, we show that for causal adversaries, for a large range of $p$ (for all $p > 0.25$), the maximum achievable rate equals that of the classical ``omniscient" adversarial model ({\em i.e.}, $0$).
This may at first come as a surprise, as the online adversary is weaker than the omniscient one, and hence one may suspect that it allows a higher rate of communication.

We have two main results. Theorem~\ref{the:det} gives an upper bound on the weak capacity $\Ca^{\tt{w}}$
if Alice's encoder is deterministic.
Theorem~\ref{the:prob} gives an upper bound on the strong capacity
$\Ca^{\tt{s}}$ in the more general case where Alice's encoder is probabilistic.
Due to certain limitations of our proof techniques, we do not present any bounds on the weak capacity in the latter setting.
The upper bound in both cases equals $\min\{1-H(p),(1-4p)^+\}$.


\begin{theorem}[Deterministic encoder]
\label{the:det}
For deterministic codes,  $\Ca^{\tt{s}} \leq \Ca^{\tt{w}} \leq \min\{1-H(p),(1-4p)^+\}$.
\end{theorem}


\begin{theorem}[Probabilistic encoder]
\label{the:prob}
For probabilistic codes,  $\Ca^{\tt{s}} \leq \min\{1-H(p),(1-4p)^+\}$.
\end{theorem}

We note that under a very weak notion of capacity in which one only requires the success probability to be bounded away from zero (instead of approaching $1$), the capacity of the omniscient channel, and thus the binary causal-adversary channel, approaches $1-H(p)$. This follows by the fact that for $\bl$ sufficiently large and $\ell \geq 4$ there exists $(n,Rn)$ codes which are $(\ell,pn)$ {\em list decodable} with $R=1-H(p)(1+1/\ell)$ \cite{Elias91}. Communicating using an $(\ell,pn)$ list decodable code allows Bob to decode a list of size $\ell$ of messages which includes the message transmitted by Alice. Choosing a message uniformly at random from his list, Bob decodes correctly with probability at least $1/\ell$.


\subsection{Outline of proof techniques}\label{sec:outline}
The upper bound of $1-H(p)$ follows directly by describing an attack for Calvin wherein he approximately simulates a BSC($p$) (Binary Symmetric Channel~\cite{CT06} with crossover probability $p$). More precisely, for each $i \in [n]$ and any sufficiently small $\e >0$, Calvin flips $x_i$ with probability $p-\e$ until he runs out of his budget of $p\bl$ bit-flips. By the Chernoff bound~\cite{chernoff1}, with very high probability
he does not run out of his budget, and is therefore indistinguishable from a BSC($p-\e$). But it is well-known~\cite{CT06} that in this case the optimal rate of communication from Alice to Bob is $1-H(p-\e)$.
Taking the limit when $\e \rightarrow 0$ implies our bound.

The upper bound of $(1-4p)^+$ is more involved.
For the case where Alice's encoder is deterministic, the proof of Theorem~\ref{the:det} has the following overall structure.
Assume for sake of contradiction that Alice attempts to communicate at rate greater than $\rate = (1-4p)^+$.
To prove our upper bound we design the following {\em wait-and-push} attack for Calvin.

Calvin starts by waiting for Alice to transmit approximately $\rate \bl$ bits.
As Alice is assumed to communicate at rate greater than $\rate$, the set of Alice's codewords $\cX'$ {\em consistent} with the bits Calvin has seen so far is ``large" with ``high probability".
Calvin constructs $\cX'$ and chooses a codeword $\bx'$ uniformly at random from $\cX'$.
He then actively ``pushes" $\bx$ in the direction of $\bx'$ by flipping, with probability $1/2$, each future $x_i$ that differs from $x_i'$. 
If Calvin succeeds in pushing $\bx$ to a word  $\by$ roughly midway between $\bx$ and $\bx'$, a careful analysis demonstrates that regardless of Bob's decoding strategy, Bob is unable to determine whether Alice transmitted $\bx$ or $\bx'$ --- causing a decoding error of $1/2$ in this case.
So, to prove our bound, we must show that with constant probability (independent of the block length $\bl$) Calvin will indeed succeed in pushing $\bx$ to $\by$.
Namely, that Alice's codeword $\bx$ and the codeword chosen at random by Calvin $\bx'$ are of distance at most $2p\bl$.
Roughly speaking, we prove the above by a detailed analysis of the distance structure of the set of codewords in any code using tools from extremal combinatorics and coding theory.

The case where Alice's encoder may be randomized is more technically challenging, and is considered in Theorem~\ref{the:prob}. 
At a high level, the strategy of Calvin for a probabilistic encoder follows that outlined for the deterministic case.
However, there are two main difficulties in its extended analysis. 
Firstly, the symmetry between $\bx$ and $\bx'$ no longer exists.
Namely, the fact that Bob may not be able to distinguish which of the two were transmitted by Alice does not necessarily cause a significant decoding error, since the probability of $\bx'$ being transmitted by Alice may well be significantly smaller than the probability that $\bx$ was transmitted.
Secondly, the fact that both $\bx$ and $\bx'$ may correspond to the same message $u$ places the entire scheme in jeopardy.
As it now no longer matters if Bob decodes to $\bx$ or $\bx'$, in both cases the decoded message will be that sent by Alice.

To overcome these difficulties, we describe a more intricate analysis of Calvin's attack.
Roughly speaking, we prove that a ``large" subset $\cX''$ of $\cX'$ behaves ``well". 
Any $\bx'$ chosen uniformly at random from $\cX'$, with ``significant" probability, is in $\cX''$, and has three properties corresponding to those when Alice uses a deterministic encoder. That is, $\bx'$ is sufficiently {\em close} to $\bx$ as desired, it has approximately the same probability of transmission that $\bx$ does (thus preserving the needed symmetry), and it also corresponds to a message that differs from that corresponding to $\bx$.
All in all, we show that the above three properties hold with probability $1/\poly(n)$, which suffices to bound the strong capacity of the channel at hand (but not the weak capacity).

In case of a randomized encoder of Alice, we assume that the messages 
may have nonuniform distribution, and also any message is encoded into
one of a set of possible codewords as per some probability distribution
in that set. One may think of various other ways of encoding, for example
the following, to confuse Calvin. But as we discuss in the next
paragraph, such schemes are also covered in our setup.

{\em Multiple codebooks:} In this scheme, Alice maintains a set of
codes $\cC_1, \cC_2, \ldots, \cC_L$. For transmitting a message $u$, she
randomly selects the code $\cC_i$ with probability $q_i$. If the set
of messages is $\cU = \{1,2,\ldots, M\}$ with a probability distribution
given by $p_i \defn Pr \{u = i\}$, and the code $\cC_r$ contains
the codewords $\{\bx (u,r) \mid u=1,2,\ldots,M\}$, then in our setup,
the corresponding codebook for the message $u$ will be
 $\cX (u) = \{\bx (u, r) \mid r=1,2,\ldots, L\}$.
This codebook may have less than $L$ codewords
due to common codewords in the original codes. The induced probability
distribution in this codebook of $u$ is given by $Pr\{\bx|u\}
= \sum_{r: \bx(u,r) = \bx} q_r$. 

If Alice picks a code and uses it to encode several messages, even
then she does not gain anything. First, if she uses the same code
to encode too many messages (and calvin knows the encoding scheme,
as assumed), then both Bob and Calvin will know the code used after
receiving or `reading' some codewords. On the other hand, if a 
randomly chosen code is used only to encode a block of 
few messages this is equivalent to using a longer (`superblock') code
in our setup. The only difference is that the probability of error
analysed in our set up is the probability of error in decoding
the `superblocks' rather than the smaller blocks/codewords.

The proofs of the upper bounds corresponding to $1-H(p)$ have already been sketched in Section~\ref{sec:outline}.
Hence we only provide proofs of the upper bounds corresponding to $(1-4p)^+$ in Theorems~\ref{the:det} and~\ref{the:prob}.

\section{Proof of Theorem~\ref{the:det}}

Let $\rate = (1-4p)^++\e$ for some $\e >0$. Let $\log(.)$ denote the binary logarithm, here and throughout.
By assumption for deterministic codes, Alice's message space $\cU$ is of size $2^{\rate \bl}$.
Here we assume for that $2^{\rate \bl}$ in an integer.
This implies that the set $\cX$ of Alice's transmitted codewords is of size $2^{\rate \bl}$. 
\footnote{In fact, $\cX$ may be smaller, however we note that for codes of optimal rate, $|\cX|$ is of size {\em exactly} $2^{\rate \bl}$. If $|\cX| < 2^{Rn}$, then for some transmitted codeword $\bx$ at least two messages $\mess$ and $\mess'$ must both be encoded to $\bx$. On receiving $\bx$, Bob's probability of error is maximal -- it is at least $1/2$. Therefore changing the codebook so as to encode $\mess'$ as some $\bx' \notin \cX$ cannot increase the probability of decoding error.}

We now present Calvin's attack. We show that for any fixed $\e > 0$, regardless of Bob's decoding strategy, there is a decoding error with constant probability (namely, the error probability is independent of $n$).
Calvin's attack is in two stages.
First Calvin {\em passively} waits until Alice transmits $\ell=(R-\e/2)n$ bits over the channel.
Let $\bx^\ell \in \{0,1\}^{\ell}$ be the value of the codeword observed so far.
He then considers the set of codewords that are consistent with
the observed $\bx^\ell$.
Namely, Calvin constructs the set  $\cX|_{\bx^\ell} = \{\bx=x_1,\ldots, x_n \in \cX \mid x_1,\ldots, x_\ell = \bx^\ell\}$. He then
chooses an element ${\bf x'} \in \cX|_{\bx^\ell}$
uniformly at random.
In the second stage, Calvin follows a {\em random bit-flip strategy}. That is, for each remaining bit $x_i'$ of ${\bf x'}$ that differs from the
corresponding bit $x_i$ of ${\bf x}$ transmitted, he flips the transmitted bit
with probability $1/2$, until he has either flipped $p \bl$ bits, or until $i=\bl$.

We analyze Calvin's attack by a series of claims.
We first show that with high probability (w.h.p.) the set $\cX|_{\bx^\ell}$ is {\em large}.

\begin{claim}
With probability at least $1-2^{- \e n/4}$, the set $\cX|_{\bx^\ell}$ is of size at least $2^{\e \bl/4}$.
\label{claim:suff_set}
\end{claim}

\begin{proof}
The number of messages $u$ for which $\cX|_{\bx^\ell(\mess)}$ is of size less than $2^{\e n/4}$ is at most the number of  distinct prefixes $\bx^\ell$ times $2^{\e n/4}$, which in turn is at most $2^{\ell + \e n/4} = 2^{(R-\e/4)n}$.
\end{proof}

Now assume that the message $\mess$ is such that its corresponding set $\cX|_{\bx^\ell (\mess)}$ is of size at least $2^{\e n/4}$.
We now show that this implies that the transmitted codeword $\bx$ and the codeword $\bx'$ chosen by Calvin are distinct and of {\em small} Hamming distance apart
with a positive probability (independent of $n$).

\begin{claim}
Conditioned on Claim~\ref{claim:suff_set}, with probability at least $\frac{\e}{64p}$, $\bx \neq \bx'$ and $d_H(\bx,\bx') < 2pn- \e n/8$.
\label{claim:hamm}
\end{claim}

\begin{proof}
Consider the undirected graph $\Graph=(\cV,\cE)$ in which the vertex set $\cV$ consists of the set $\cX|_{\bx^\ell}$ and  two nodes are connected by an edge if their Hamming distance is less than $d=2pn- \e n/8$.
An {\em independent set $\cI$ in $\Graph$} corresponds to a subset of codewords in $\{0,1\}^n$ that are all (pairwise) at distance greater than $d$.

Since the codewords in $\cX|_{\bx^\ell}$ all have the same prefix $\bx^\ell$, one may consider only the suffix (of length $\bl-\ell = 4p\bl - \e \bl/2$) of the codewords in $\cX|_{\bx^\ell}$. Here we assume $p \leq 0.25$, minor modifications in the proof are needed for larger $p$. The set of vectors defined by the suffixes in an independent set $\cI$ of $\Graph$ now corresponds to a binary error-correcting code of length $4p\bl - \e \bl/2$, with $|\cI|$ codewords and minimum distance $d$.	

By Plotkin's bound~\cite{brouwer1} there do not exist binary error correcting codes with more than $\frac{2d}{2d-(4p\bl - \e \bl/2)}+1$ codewords.
Thus $\cI$, any maximal independent set in $\Graph$, must satisfy
\begin{equation}
|\cI| \leq \frac{2(2pn-\e n/8)}{2(2pn -\e n/8)-4pn + \e n/2}+1 = \frac{16p}{\e}
\label{eq:ind}
\end{equation}

By Tur\'{a}n's theorem \cite{turan}, any undirected graph $\Graph$ of size $|\cV|$ and average degree $\Delta$ has an independent set of size at least $|\cV|/(\Delta+1)$. This, along with (\ref{eq:ind}) implies that the average degree of our graph $\Graph$ satisfies
$$
\frac{|\cV|}{\Delta+1}\leq |\cI| \leq \frac{16p}{\e}
$$
This in turn implies that
$$
\Delta \geq \frac{\e |\cV|}{16 p} -1 \geq \frac{\e |\cV|}{32p}
$$
The second inequality is for large enough $\bl$,
since $|\cV|$ is of size at least $2^{\rate \bl}$.
To summarize the above discussion, we have shown that our graph $G$ has {\em large} average degree of size $\Delta \geq \frac{\e |\cV|}{32p}$. We now use this fact to analyze Calvin's attack.

By the definition of deterministic codes, any codeword in $\cX$ is transmitted with equal probability.
Also, by definition both $\bx$ (the transmitted codeword) and $\bx'$ (the codeword chosen by Calvin) are in $\cV=\cX|_{\bx^\ell}$. Hence both $\bx$ and $\bx'$ are uniform in $\cX|_{\bx^\ell}$. This implies that with probability $|\cE|/|\cV|^2$ the nodes corresponding to codewords $\bx$ and $\bx'$ are distinct and connected by an edge in $\Graph$.
This in turn implies that with probability $|\cE|/|\cV|^2$, $\bx \neq \bx'$ and $d_H(\bx,\bx') < 2p\bl- \e \bl/8$, as required.
Now
$$
\frac{|\cE|}{|\cV|^2} = \frac{\Delta|\cV|}{2|\cV|^2} \geq  \frac{\e}{64p}
$$
\end{proof}

Conditioned on Claim~\ref{claim:hamm}, Calvin's codeword $\bx'$ is very close to Alice's transmitted codeword $\bx$.
Specifically, $d_H(\bx,\bx') \in (0,2pn- \e n/8)$.
We now show that if Calvin follows the random bit-flip strategy, from Bob's perspective (w.h.p.), both $\bx$ or $\bx'$ were equally likely to have been transmitted by Alice.



We first show that during Calvin's random bit-flip process, w.h.p., Calvin does not ``run out'' of his budget of $p\bl$ bit flips.

\begin{claim}
Conditioned on Claim~\ref{claim:hamm}, with probability at least $1-2^{-\Omega(\e^2 \bl)}$
\begin{equation}
d_H(\bx,\by) \in \left(\frac{d}{2}-\frac{\e \bl}{16},\frac{d}{2}+\frac{\e \bl}{16}\right).
\label{eq:no_bit_flips}
\end{equation}
\label{claim:budget}
\end{claim}

\begin{proof}
The expected number of locations flipped by Calvin is $d/2 \leq p\bl-\e \bl/16$.
Assume that $d/2 = p\bl-\e \bl/16$ (for smaller values of $d$ the bound is only tighter).
By Sanov's theorem~\cite[Theorem 12.4.1]{CT06}, the probability that the number of bits flipped by Calvin deviates from the expectation $d/2$ by more than $\e \bl/16$ is at most $e^{-\Omega(\e^2 n^2/d)} \leq e^{-\Omega(\e^2 n)}$
for large enough $\bl$.
\end{proof}

It should be noted that $d/2+ \e n/16 \leq pn$, and so
$d_H(\bx,\by) \leq d/2+\e \bl/16$ implies that the number
of bits flipped by Calvin does not exceed $pn$.
Since Calvin possibly flips only the bits of $\bx$ which differ
from the corresponding bits in $\bx^\prime$, (\ref{eq:no_bit_flips}) also implies
\begin{equation}
d_H(\bx^\prime,\by) \in \left(\frac{d}{2}-\frac{\e \bl}{16},\frac{d}{2}+\frac{\e \bl}{16}\right).
\end{equation}

We conclude by proving that if the number of bits flipped by Calvin
lies in the range $(d/2-\e \bl/16, d/2+\e \bl/16)$, then indeed Bob cannot
distinguish between the case in which $\bx$ or $\bx'$ were transmitted.

\begin{claim}
Conditioned on Claim~\ref{claim:budget} Bob makes a decoding error with probability at least $1/2$.
\label{claim:error}
\end{claim}
\begin{proof}
By Bayes' Theorem~\cite{feller1}, if Bob receives $\by$, the {\it a posteri}
probability that Alice  transmitted $\bx$, denoted $p(\bx | \by)$, equals $p(\by| \bx)p(\bx)/p(\by)$. Here $p(\bx)$ is the probability (over her encoding strategy) that Alice transmits $\bx$, $p(\by | \bx)$ is the probability (over Calvin's random bit-flipping strategy) that Bob receives $\by$ given that Alice transmits $\bx$, and $p(\by)$ is the resulting probability that Bob receives $\by$.
Similarly, $p(\bx' | \by)=p(\by| \bx')p(\bx')/p(\by)$.
Taking the ratio and noting that for deterministic codes $p(\bx)= p(\bx')$, we have
\begin{equation}
p(\bx | \by)/p(\bx' | \by)=p(\by| \bx)/p(\by| \bx').
\label{eq:cond_prob}
\end{equation}

Since Calvin's random bit-flip strategy involves him flipping bits of
$\bx$ (which are different from the corresponding bits of $\bx^\prime$) with probability $1/2$, for all $\by$ satisfying (\ref{eq:no_bit_flips}),
the probabilities $p(\by| \bx)$ and $p(\by| \bx')$ are equal. This observation and (\ref{eq:cond_prob}) together imply $p(\bx | \by) = p(\bx' | \by)$. 
Thus, Bob cannot distinguish whether $\bx$ or $\bx'$ were transmitted.
Namely, on the pair of events in which Alice transmits $\bx$ and Calvin chooses $\bx'$ and in which Alice transmits $\bx'$ and Calvin chooses $\bx$, no matter which decoding process Bob uses, he will have an average decoding error of at least $1/2$. This suffices to prove our assertion.
\end{proof}

Thus a decoding error happens if the conditions of Claims~\ref{claim:suff_set},~\ref{claim:hamm},~\ref{claim:budget} and~\ref{claim:error} are all satisfied. This happens with probability at least
$\left( 1-2^{- \e n/4} \right ) \left( \frac{\e}{64p} \right ) \left (1-2^{-\Omega(\e^2 \bl)} \right ) \left( \frac{1}{2} \right ) \geq \left( \frac{1}{2} \right ) \left( \frac{\e}{64p} \right ) \left (\frac{1}{2} \right ) \left( \frac{1}{2} \right ) \geq \frac{\e}{512p} $ for large enough $\bl$.

\hfill $\blacksquare$

\section{Proof of Theorem~\ref{the:prob}}

We start by proving the following technical Lemma that we use in our proof.
Let $q$ be an arbitrary probability distribution over an index set $I=\{1,\dots,k\}$. 
Let $\mathbf{A_1},\ldots,\mathbf{A_{k}}$ be arbitrary discrete random variables with probability distributions $q_1,\ldots,q_{k}$ over alphabets $\cA_1,\ldots,\cA_{k}$ respectively. Let $k_i = |\cA_i|$. Let $\mathbf{A}$ be a random variable that equals the random variable $\mathbf{A_i}$ with probability $q(i)$. Then the  following Lemma describing an elementary property of the entropy function $H(.)$ is useful in the proof of Theorem~\ref{the:prob}.

\begin{lemma} The entropies of $\mathbf{A}, \mathbf{A_1},\ldots,\mathbf{A_k}$ and $q$ satisfy
$H(\mathbf{A}) \leq \sum_{i = 1}^k q(i) H(\mathbf{A_i}) + H(q)$, with equality if and only if for each $i,i'$ for which both $q(i)$ and $q(i')$ are positive it holds that $\Pr _{q_i,q_{i'}}[\mathbf{A_i} =  \mathbf{A_{i'}}] = 0$.
\label{lem:chain_rule}
\end{lemma}

\begin{proof} 
For any $a \in \cA$, the probability $\Pr\{\mathbf{A} = a\}=p(a)$ of occurrence of $a$, equals $\sum_{i: a \in \cA_i}q(i)q_i(a)$. 
Hence
\begin{eqnarray}
H(\mathbf{A}) & = & - \sum_{a \in \bigcup_i\cA_i} p(a) \log(p(a))\nonumber \\
& \leq & -\sum_{i=1}^k\sum_{j=1}^{k_i}q(i)q_i(j) \log(q(i)q_i(j)) \label{eq:jensen}\\
  & =  & \sum_{i=1}^k\sum_{j=1}^{k_i}q(i) \left ( q_i(j) \log(q_i(j)) \right ) \nonumber  \\
   &\  & \ \ \ \ \ \ \ \ \ \ + \sum_{i=1}^k\sum_{j=1}^{k_i} q_i(j) \left ( q(i) \log(q(i)) \right ) \nonumber \\
  & = &\sum_{i = 1}^k q(i) H(\mathbf{A_i}) + H(q). \nonumber
\end{eqnarray}
Here (\ref{eq:jensen}) follows from Jensen's inequality, e.g.~\cite{CT06}, with equality if and only if for each positive $\Pr\{\mathbf{A} =a\}$, there is a unique 
$i$ such that $q(i)q_i(j) > 0$ (here $a_i(j)=a$).
\end{proof}

We now turn to prove Theorem~\ref{the:prob}.
Recall our notation:
let $\bU$ be the random variable corresponding to Alice's
message and $p_U$ its distribution (with entropy $
\rate \bl$). Throughout we assume the message set $\cU$ (the support of $\bU$) is at most of size $2^n$.
Let  $\cX$ be Alice's codebook. 
$\cX$ is a collection $\{\cX(\mess)\}$ of subsets of
$\{0,1\}^\bl$. For each subset $\cX(\mess) \subset \cX$, there is a
corresponding codeword random variable $\bX(\mess)$ with
codeword distribution $p_{X(\mess)}$ over $\cX(\mess)$.
For any value $\bU=\mess$ of the message, Alice's encoder choses a
codeword from $\cX(\mess )$ randomly from the distribution
$p_{X(\mess)}$. 
Alice's message distribution $p_U$, codebook $\cX$, and all the codebook
distributions $p_{X(\mess)}$ are all known to both Bob and Calvin,
but the values of the random variables $\bU$ and $\bX(.)$ are
unknown to them. If $\cX(\mess) = \{\bx (\mess, r) : r \in \Lambda_\mess\}$,
then the transmitted
codeword $\bX(\bU)$ has the probability distribution given by
$\Pr [\bX(\bU) = \bx (\mess, \ind)] = p_U(u)p_{X(\mess)}(\bx(\mess,\ind))$.
Let $p$ the the overall distribution of codewords $\bx = \bx(\mess,r)$ of Alice.
It holds that $p(\bx(\mess,\ind))=p_U(u)p_{X(\mess)}(\bx)$ and 
$p(\bx)=\sum_{\cU}{p_U(u)p_{X(\mess)}(\bx)}$.

For any $\e>0$, let $\rate = (1-4p)^++\e$.
We start by specifying Calvin's attack. 
Calvin uses a very similar attack to the one described in the proof of Theorem~\ref{the:det}.
That is, Calvin first {\em passively} waits until Alice transmits $\ell=(R-\e/2)n$ bits over the channel.
Let $\bx^\ell \in \{0,1\}^{\ell}$ be the value of the codeword observed so far.
He then considers the set of codewords $\bx(\mess,\ind)$ {\em consistent} with
the observed $\bx^\ell$.
Here and throughout this section, we denote codewords by their corresponding message $\mess$ and index $\ind$ in $\cX(u)$.
As it may be that $\bx(\mess,\ind)$ is exactly the same codeword as $\bx(\mess',\ind')$, the sets in the definitions to follow and in this section are in a sense {\em multisets}.
Namely, Calvin constructs the set  $\cX|_{\bx^\ell} = \{\bx(\mess,\ind)=x_1,\ldots, x_n \in \cX \mid x_1,\ldots, x_\ell = \bx^\ell\}$. Let $p(\bx^\ell)=p(\cX|_{\bx^\ell})$ be the probability, under the probability distribution $p$, corresponding to the event that Calvin observes $\bx^\ell$ in the first $\ell$ transmissions. Let $p_{U|_{\bx^\ell}}$ 
and $p_{X(\mess)|_{\bx^\ell}}$ 
be the probability distributions $p_U$ 
and $p_{X(\mess)}$ also respectively 
conditioned on the same event.
Calvin then
chooses an element ${\bx'(\mess',\ind')} \in \cX|_{\bx^\ell}$ with probability\footnote{This is one significant difference from the attack in the proof of Theorem~\ref{the:det} -- there Calvin chooses each $\bx'$ uniformly at random from the corresponding consistent set.}
 $p_{U|_{\bx^\ell}}(\mess')p_{X(\mess')|_{\bx^\ell}}(\bx'(\mess',\ind'))$. In the second stage he then follows exactly the same {\em random bit-flip strategy} as in the proof of Theorem~\ref{the:det}.

Recall that in the proof of Theorem~\ref{the:det}, our goal was to prove that with some constant probability, the distance between $\bx(\mess,\ind)$ and $\bx'(\mess',\ind')$ is approximately $2p\bl$. Loosely speaking, this allows the success of Calvin's attack (i.e., imply a decoding error). Following the same outline of proof, we now show that with probability $1/\poly{(n)}$ the codeword $\bx'(\mess',\ind')$ chosen by Calvin has the following three properties:
\begin{itemize}
\item It's corresponding message differs from that corresponding to $\bx(\mess,\ind)$ (i.e., $\mess \ne \mess'$).
\item $\bx'(\mess',\ind')$ is {\em close} to $\bx(\mess,\ind)$ and thus Calvin will be able to ``push'' $\bx(\mess,\ind)$ to a codeword $\by$ at approximately the same distance from $\bx(\mess,\ind)$ and $\bx'(\mess',\ind')$.
\item Given $\by$, Bob is unable to distinguish whether $\bx(\mess,\ind)$ or $\bx'(\mess',\ind')$ was transmitted.
\end{itemize}
To this end, we partition the set $\cX|_{\bx^\ell}$ into $n^2$ disjoint subsets $\cX_{ij}$ for $i,j \in \{1,2,\dots,n\}$. 
Let $p(\cX_{ij})$ be the probability mass of $\cX_{ij}$.
Let $p_{U|_{ij}}$ and $p_{X(\mess)|_{ij}}$ be the probability distributions $p_U$ and $p_{X(\mess)}$ respectively conditioned on the event that Alice transmitted $\bx(\mess,\ind)$ in $\cX_{ij}$.
The partition $\cX_{ij}$ is obtained in two steps -- first we partition $\cX|_{\bx^\ell}$ into $n$ subsets $\cX_i$, then we partition each $\cX_i$ into $n$ sets $\cX_{ij}$. We also use the probability distribution $p(\cX_i)$, $p_{U|_{i}}$ and $p_{X(\mess)|_{i}}$ defined accordingly.
All in all, we prove the existence of a subset $\cX_{ij}$ with the following properties
\begin{itemize}
\item $H(p_{U|_{ij}})$ is ``large".
\item $p(\cX_{ij})$ is large with respect to $p(\bx^\ell)$.
\item For any $\bx(\mess,\ind) \in \cX_{ij}$ it holds that $p(\bx(\mess,\ind))$ has approximately  the same value.
\item $p_{U|_{ij}}$ is approximately uniform on its support.
\end{itemize}
Roughly speaking, proving these properties on $\cX_{ij}$ reduces us to the case of a deterministic encoder (addressed in Theorem~\ref{the:det}) and allows us to complete our proof.

We now present our proof for the existence of $\cX_{ij}$ as specified above.
We first show that with positive probability the set $\cX|_{\bx^\ell}$ has {\em high} entropy.

\begin{claim}
\label{claim:entropy1}
With probability at least $\e/4$,  $H(p_{U|_{\bx^\ell}}) \geq \e n/4$.
\end{claim}

\begin{proof} 
Let $q$ be the probability distribution over $\{0,1\}^\ell$ for which $q(\bx^\ell) = p(\bx^\ell)$ for all possible $\bx^{\ell} \in \{0,1\}^\ell$. Let $q_{\bx^\ell}$ be the probability distribution $p_{U|_{\bx^\ell}}$. Now using Lemma~\ref{lem:chain_rule} we obtain
\begin{equation}
H(p_U) \leq \sum_{\bx^\ell} q(\bx^\ell) H(p_{U|_{\bx^\ell}}) + H(q).
\label{eq:cr1a}
\end{equation}
By our definitions $H(p_U) = \rate \bl$.
Moreover, $H(q) \leq \ell = (\rate-\e/2)\bl$ (since $q$ is defined over an alphabet of size $2^{\ell}$). 
Thus (\ref{eq:cr1a}) becomes
$$
\sum_{\bx^\ell} q(\bx^\ell) H(p_{U|_{\bx^\ell}}) \geq \rate \bl - (\rate-\e/2)\bl = \e\bl/2.
\label{eq:cr1}
$$
As the average of $H(p_{U|_{\bx^\ell}})$ is at least $\e\bl/2 $, then $H(p_{U|_{\bx^\ell}}) \geq \e \bl/4$ with probability at least $\e/4$ (by a Markov type inequality, here we use the fact that $H(p_{U|_{\bx^\ell}}) \leq \bl$).
\end{proof}

We now define the sets $\cX_i$.
For $i = 1,\dots,n-1$, let $\cX_i$ be the set of codewords in $\cX|_{\bx^\ell}$ for which $p(\bx(\mess,\ind))/p(\bx^\ell)$ is in the range $(2^{-3i},2^{-3i+3}]$.
The set $\cX_n$ is defined to be the set of codewords in $\cX|_{\bx^\ell}$ for which $p(\bx(\mess,\ind))/p(\bx^\ell)$ is in the range $[0,2^{-3n+3}]$.
Let $p(\cX_i)$ be the probability mass of $\cX_i$.
Namely $p(\cX_i) \simeq 2^{-3i}|\cX_i|p(\bx^\ell)$.
Let $q$ be the distribution over $\{1,2,\dots,n\}$ taking $i$ w.p. $p(\cX_i)/p(\bx^\ell)$.
Notice that $H(q) \leq \log(n) = o(n)$ (as its support is of size $n$).
Conditioning on Claim~\ref{claim:entropy1} and using Lemma~\ref{lem:chain_rule} it can be verified that 

\begin{claim}
\label{claim:entropy2}
\begin{equation}
\label{eqn:cr2}
\sum_i{q(i)H(p_{U|_{i}})} \geq H(p_{U|_{\bx^\ell}}) - H(q) \geq \e \bl/8
\end{equation}
\end{claim}

Consider sets $\cX_i$ with (relative) mass $q(i) \geq 1/n^2$.
It holds that
$$
\sum_{i \leq n-1;q(i)\geq 1/n^2} {q(i)H(p_{U|_{i}})} \geq \e n/16
$$
The above follows from the fact that
$
\sum_{i \leq n-1;q(i)\leq 1/n^2} {q(i)H(p_{U|_{i}})} + q(n)H(p_{U|_{i}}) \leq \sum_{i \leq n-1;q(i)\leq 1/n^2} {n/n^2} + 2^{-n+
3}n \leq 2$ (for sufficiently large $n$). Here we use the fact that $q(n) \leq |\cX_i| 2^{-3n+3}$.

We conclude the existence of a set $\cX_i$ such that $q(i) \geq 1/n^2$ and $H(p_{U|_{i}}) \geq \e n /16$.
We now further partition $\cX_i$.
For $j= 1,\dots,n-1$, let $\cX_{ij}$ be the set of codewords $\bx(\mess,\ind)$ in $\cX_i$ for which $p_{U|_i}(\mess)$ is in the range $(2^{-3j},2^{-3j+3}]$.
$\cX_{in}$ is defined to be the set of codewords $\bx(\mess,\ind)$ in $\cX_i$ for which $p_{U|_i}(\mess)$ is in the range $[0,2^{-3n+3}]$. 
Let $p(\cX_{ij})$ be the probability mass of $\cX_{ij}$.
Namely $p(\cX_{ij}) \simeq 2^{-3i}|\cX_{ij}|p(\bx^\ell)$.
Let $q'$ be the distribution over $\{1,2,\dots,n\}$ taking $j$ w.p. $p(\cX_{ij})/p(\cX_i)$.
Notice that $H(q') \leq \log(n) = o(n)$ (as its support is of size $n$).
As before, conditioning on Claim~\ref{claim:entropy2} and using Lemma~\ref{lem:chain_rule} it can be verified that (for the index $i$ specified above),

\begin{claim}
\label{clain:entropy2}
\begin{equation}
\label{eqn:cr2}
\sum_j{q'(j)H(p_{U|_{ij}})} \geq H(p_{U|_i}) - H(q') \geq \e \bl/32
\end{equation}
\end{claim}

Again, consider sets $\cX_{ij}$ with mass $q'(i) \geq 1/n^2$.
It holds that
$$
\sum_{j \leq n-1;q'(j)\geq 1/n^2} {q'(j)H(p_{U|_{ij}})} \geq \e n/64
$$

We conclude the existence of a set $\cX_{ij}$ such that 
\begin{itemize}
\item $H(p_{U|_{ij}}) \geq  \e n/64$.
\item $p(\cX_{ij}) \geq p(\bx^\ell)/n^4$.
\item For any $\bx(\mess,\ind) \in \cX_{ij}$ it holds that $p(\bx(\mess,\ind))$ is approximately $2^{-3i}p(\bx^\ell)$.
\item For any $\bx(\mess,\ind) \in \cX_{ij}$ it holds that $p_{U|_{ij}}(\mess)$ is approximately equal.
\end{itemize}

The set $\cX_{ij}$ is exactly what we are looking for.
Roughly speaking, by Claim~\ref{claim:entropy1}, with probability at least $\e/4$ Calvin views a prefix $\bx^\ell$ for which $H(p_{U|_{\bx^\ell}}) \geq \e n/4$.
Conditioning on this event, both Alice and Calvin choose codewords $\bx(\mess,\ind)$, $\bx'(\mess',\ind')$ in $\cX_{ij}$ with probability at least $1/n^8$.

We now sketch to remainder of the proof which closely follows that of Theorem~\ref{the:det}.
We partition $\cX_{ij}$ into {\em groups} of messages $\cX_{ij}(\mess)$ consisting of all codewords in $\cX_{ij}$ corresponding to $\mess$. 
Recall that each codeword $\bx(\mess,\ind) \in \cX_{ij}$ has approximately the same probability $p(\bx(\mess,\ind))$, and for each $\bx(\mess,\ind) \in \cX_{ij}$ it holds that $p_{U|_{ij}}(\mess)$ is approximately the same value.
This implies that each group $\cX_{ij}(\mess) \subseteq \cX_{ij}$ has approximately the same size.
Moreover, as $H(p_{U|_{ij}}) \geq  \e n/64$ it holds that there are at least $2^{\e n/64}$ non-empty subsets $\cX_{ij}(\mess)$ in  $\cX_{ij}$.

So, all in all, $\cX_{ij}$ has a very symmetric structure: it includes {\em many} groups, each consisting of elements with the same transmission probability, and each of approximately the same size and mass (w.r.t. $p$).
This reduces us to the case considered in Theorem~\ref{the:det} in which our subset $\cX|_{\bx^\ell}$ included many messages, each with the same probability, details follow.

Consider the graph $\Graph=(\cV,\cE)$ in which the vertex set $\cV$ consists of the set $\cX_{ij}$ and  two nodes are connected by an edge if their Hamming distance is less than $d=2pn- \e n/8$.

Now, it is can be verified (using analysis almost identical to that given in the proof of Theorem~\ref{the:det})
that
\begin{enumerate}
\item With probability at least $1-2^{-\Omega(\e \bl)}$ the codewords $\bx(\mess,\ind)$ and $\bx'(\mess',\ind')$ satisfy $\mess \ne \mess'$. Here one needs to take into consideration the slight difference in the group sizes and the probabilities for each codeword.
\item With probability $\Omega \left (\frac{\e}{p} \right )$ the vertices in $\Graph$ corresponding to $\bx(\mess,\ind)$ and $\bx'(\mess',\ind')$ are connected by an edge.
\item During Calvin's random bit-flip process, with high probability of $1-2^{-\Omega(\e^2 \bl)}$, Calvin does not ``run out'' of his budget of $p\bl$ bit flips.
\item Conditioning on the above, Bob cannot distinguish between the case in which $\bx(\mess,\ind)$ or $\bx'(\mess',\ind')$ were transmitted.
\item Finally, on the pair of events in which Alice transmits $\bx(\mess,\ind)$ and Calvin chooses $\bx'(\mess',\ind')$, and Alice transmits $\bx'(\mess',\ind')$ and Calvin chooses $\bx(\mess,\ind)$, no matter which decoding process Bob uses, he has an average decoding error that is bounded away from zero. Here again we take into account the slight differences between $p(\bx(\mess,\ind))$ and $p(\bx'(\mess',\ind'))$.
\end{enumerate}

To summarize, Calvin causes a decoding error with probability $\Omega(\poly(\e)/\poly(n))=\Omega(1/\poly(n))$ as desired. This concludes our proof.
\hfill $\blacksquare$


\section{Conclusions}
We analyze the capacity of the causal-adversarial channel and show (for both deterministic and probabilistic encoders) that the capacity is bounded by above by $\min\{1-H(p),(1-4p)^+\}$. For a large range of $p$ (for all $p > 0.25$), the maximum achievable rate equals that of the {\em stronger} classical ``omniscient" adversarial model ({\em i.e.}, $0$).

Several questions remain open. 
In this work we do not address achievability results (i.e., the construction of codes).
It would be very interesting to obtain codes for the causal-adversary channel which obtain rate greater than that know for the ``omniscient" adversarial model ({\em i.e.}, the Gilbert-Varshamov bound) for $p < 0.25$).
As we do not believe that the upper bound of $(1-4p)^+$ presented in this work is actually tight, such codes, if they exist, may give a hint to the correct capacity.

As done in our work on large alphabets \cite{djl}, one may also consider the more general channel model in which for a {\em delay} parameter $d \in (0,1)$, the jammer's decision on the corruption of $x_i$ must depend solely on $x_j$ for $j \leq i - d\bl$. This might correspond to the scenario in which the error transmission of the adversarial jammer is
delayed due to certain computational tasks that the adversary needs to perform. 
The capacity of the causal channel with delay is an intriguing problem left open in this work.

\end{document}